\documentclass[11pt]{article} 

\usepackage{amsmath}
\usepackage{amsthm}
\usepackage{amsfonts}
\usepackage{amssymb}
\usepackage[utf8]{inputenc}
\usepackage{geometry} 
\usepackage{graphicx} 
\usepackage{algorithm}
\usepackage{algorithmic}
\usepackage[T1]{fontenc}
\usepackage[utf8]{inputenc}
\usepackage{authblk}
\usepackage{booktabs}
\usepackage{array} 
\usepackage{paralist}
\usepackage{verbatim}
\usepackage{subfig} 
\usepackage{tikz}
\usepackage{pgfplots}

\usepackage{fancyhdr} 
\usepackage{sectsty}
\allsectionsfont{\sffamily\mdseries\upshape}

\usepackage[nottoc,notlof,notlot]{tocbibind} 
\usepackage[titles,subfigure]{tocloft}

\newtheorem{theorem}{Theorem}[section]
\newtheorem{assumption}{Assumption}
\newtheorem{proposition}[theorem]{Proposition}
\newtheorem{corollary}[theorem]{Corollary}
\newtheorem{lemma}[theorem]{Lemma}

\theoremstyle{definition}
\newtheorem{definition}[theorem]{Definition}
\newtheorem{example}[theorem]{Example}

\newcommand{\GL}{\textnormal{GL}}
\newcommand{\F}{\mathbb{F}}
\newcommand{\E}{\mathbb{E}}
\newcommand{\rk}{\textnormal{rk}}

\newcommand{\wt}{\textnormal{wt}}

\newcommand{\Gab}{\textnormal{Gab}}
\newcommand{\crk}{\textnormal{colrk}}
\newcommand{\Gr}{\textnormal{Gr}}
\newcommand{\supp}{\textnormal{supp}}

\title{Extension of Overbeck's Attack for Gabidulin-based Cryptosystems\footnote{This work was supported by SNF grant no.\ 149716.}}

\author[1]{Anna-Lena Horlemann-Trautmann}
\author[2]{Kyle Marshall}
\author[2]{Joachim Rosenthal}
\affil[1]{EPF Lausanne, Switzerland}
\affil[2]{Universit\"{a}t Z\"{u}rich, Switzerland}

\begin{document}
\maketitle

\begin{abstract}
   Cryptosystems based on codes in the rank metric were introduced in 1991 by Gabidulin, Paramanov, and Tretjakov (GPT) and have been studied as a promising alternative to cryptosystems based on codes in the Hamming metric. In particular, it was observed that the combinatorial solution for solving the rank analogy of the syndrome decoding problem appears significantly harder. Early proposals were often made with an underlying Gabidulin code structure. Gibson, in 1995, made a promising attack which was later extended by Overbeck in 2008 to cryptanalyze many of the systems in the literature. Improved systems were then designed to resist the attack of Overbeck and yet continue to use Gabidulin codes. In this paper, we generalize Overbeck's attack to break the GPT cryptosystem for all possible parameter sets, and then generalize the attack to cryptanalyze particular variants which explicitly resist the attack of Overbeck. 
\end{abstract}

\section{Introduction}

Cryptosystems based on the hardness of the general decoding problem have received much attention because of their applications to post-quantum cryptography. The practical implementation of these systems, however, has suffered from the drawback of having a large key size relative to RSA or elliptic curve cryptography (ECC). Regardless, coding based cryptography remains one of the most feasible alternatives to traditional number theoretic cryptosystems for resisting quantum attacks such as Shor's factoring algorithm. A large body of work has been produced in the study of cryptography based on codes in the Hamming metric, starting with McEliece in 1978 \cite{McEcliece1978}. It was observed that the cryptosystem he designed had efficient encryption and decryption procedures, however, the proposed public key sizes were significantly larger than keys for RSA or ECC, rendering the system infeasible in its original form. 

The large size of the key in the McEliece cryptosystem is a consequence of the efficiency of combinatorial solutions to the general decoding problem for codes in the Hamming metric. The impetus for interest in codes in the rank metric were preliminary results concerning the rank syndrome decoding problem, in which the best algorithms were of significantly higher complexity \cite{Chabaud96,Ourivski02}. This indicated that cryptosystems could be designed with far smaller parameters than those in the Hamming metric. Cryptosystems based on codes in the rank metric were introduced earlier by Gabidulin, Paramonov, and Tretjakov (GPT) \cite{Gabidulin91}. Since then, proposals for designs of cryptosystems have alternately been attacked and modified. The designs are often based on Gabidulin codes--the rank metric analogy of generalized Reed-Solomon codes--because of the scarcity of efficiently decodable codes in the rank metric. This has led to efficient structural attacks \cite{Gibson95,Overbeck08} and subsequently improvements in the designs of these codes and their parameters \cite{Kshevetskiy07,Gabidulin08,Rashwan10,Loidreau10}. It should be noted that unlike the syndrome decoding problem in the Hamming metric, the rank syndrome decoding problem is not known to be NP-hard. Other related work has been done in improving algorithms for the rank syndrome decoding problem \cite{Gaborit13}, and also designing rank-metric based cryptosystems which do not rely on Gabidulin codes \cite{Gaborit14}.

The original GPT cryptosystem had its first significant attack by Gibson \cite{Gibson95}. Following Gibson's lead, Overbeck proposed an alternative attack that led to a polynomial time break for many parameters of the GPT cryptosystem \cite{Overbeck08}. In the wake of these developments, two modifications, designed to use Gabidulin codes and yet resist the attack of Overbeck, stand out. They follow a similar idea - a more careful choice of distortion matrix - but have different approaches. The approach taken in \cite{Loidreau10} is based on enlarging the distortion matrix but restricting its rank, whereas the idea in \cite{Rashwan10} is based on careful design of the structure of the distortion matrix. While the ideas in these modifications are not necessarily mutually exclusive, an disadvantage of the former version is that it requires a large increase in the public key size in order to be secure against Overbeck's attack. A disadvantage of the latter version is that the distortion matrices must necessarily be quite structured.

In this paper we present a new attack which can be seen as a generalization of Overbeck's, and which allows us to cryptanalyze the systems presented in \cite{Rashwan10} and \cite{Loidreau10}. The paper is organized as follows: Section \ref{s:background} provides some terminology as well as the necessary background regarding rank metric codes and the GPT cryptosystem and its variants. Section \ref{s:presults} provides some basic results that we will need to describe our attack. In particular, we need some basic results about Moore matrices as well as the behavior of matrices under the coordinate-wise Frobenius map. Section \ref{s:attack} outlines the attack on the GPT cryptosystem and Section \ref{s:ext} uses the method to cryptanalyze the aforementioned variants.

\section{Background}\label{s:background}

Let $\F\subset \E$ be two fields with $[\E \colon \F] = m$. We will refer to the rank in the following ways. Given a matrix $M$ with coefficients in $\E$, we mean by the rank of $M$, the usual notion of the dimension of the row span of $M$ as a vector space over $\E$. We will denote the row span of a matrix $M$ over $\mathbb{E}$ by $\langle M \rangle$. By column rank (over $\F$) of a matrix $M$ with coefficients in $\E$, we mean the rank of the column span of $M$ as an $\F$-vector space and we will denote this by $\crk(M)$. When we say the rank of a vector, $\boldsymbol{x} \in \E^n$, we mean the $\F$-rank of the matrix obtained by expanding $\boldsymbol{x}$ into an $m \times n$ matrix according to some basis of $\E$ over $\F$. The rank defined in this way is invariant with respect to the choice of basis. An equivalent way to express the rank of a vector $\boldsymbol{x} \in \E^n$ is to take the dimension over $\F$ of the subspace of $\E$ which is spanned by the coordinates of $\boldsymbol{x}$.

If we are working over a base field $\F = \F_q$ of cardinality $q$, and an extension field $\E = \F_{q^m}$, then we denote by $[i]$ the $i$th Frobenius power, $q^{i}$. The Frobenius map can be applied to a matrix or vector coordinate-wise. If $M = (M_{a,b})$ is any matrix (or vector) over $\F_{q^m}$, we define $M^{([i])} = (M_{a,b}^{[i]})$. It is easy to verify that $\langle M\rangle^{([i])} = \{ \boldsymbol{x}^{([i])} \mid \boldsymbol{x} \in \langle M\rangle\} = \langle M^{([i])} \rangle$. 

\begin{definition} The \textit{rank distance} between $\boldsymbol{x},\boldsymbol{y} \in \F_{q^m}^n$ is defined to be $$d_{\textnormal{R}}(\boldsymbol{x},\boldsymbol{y}) = \rk (\boldsymbol{x}-\boldsymbol{y}).$$ \end{definition} This defines a metric on $\F_{q^m}^{n}$. If $V$ is a subspace of $\F_{q^m}$, the \textit{minimum rank distance} of $V$ is given by $$d_{\textnormal{R}}^{\textnormal{min}}(V) = \min\{ d_{\textnormal{R}}(\boldsymbol{x}, 0) \mid \boldsymbol{x} \in V \}.$$ We will also use the term \textit{weight} of $\boldsymbol{x}$, denoted by $\wt_{\textnormal{R}}(\boldsymbol{x})$, to mean $d_{\textnormal{R}}(\boldsymbol{x}, 0).$

The Singleton bound for (linear) rank-metric codes is given by the inequality (see e.g.\ \cite{Gabidulin85}) $$d_{\textnormal{R}}^{\textnormal{min}}(V) \leq n - \dim(V) + 1.$$ 

\begin{definition}
A rank-metric code meeting the Singleton bound is called a \emph{maximum rank-distance (MRD) code}.
\end{definition}

The linear isometries of $\F_{q^m}^n$ with respect to the rank metric are given by $(\F_{q^m}^*)\times \GL_n(\F_q)$ \cite{Berger03}.  Throughout the paper we will make extensive use of the coordinate-wise Frobenius map. This map is a semi-linear isometry of the rank-metric with many useful properties. For more information on semi-linear isometries, see \cite{Morrison13}.

\begin{definition}
	A matrix $M \in \F_{q^m}^{k\times n}$ is called a \emph{Moore matrix} if there exists a $\boldsymbol{\alpha} \in \F_{q^m}^{n}$ such that row $i$ of $M$ is equal to $\boldsymbol{\alpha}^{([i-1])}$ for $i= 1,  \dots , k$. $\boldsymbol{\alpha}$ is called the generator of $M$.
\end{definition}

In the following lemma, we will summarize some of the important properties of Moore matrices. These results are known or are direct consequences of known results.

\begin{lemma}\label{l:moore} 
Fix $k \leq N$, and let $M \in \F_{q^m}^{k\times N}$ be a Moore matrix with generator $\boldsymbol{\alpha}$, with $\rk(\boldsymbol{\alpha}) = n\leq N$. 

\begin{enumerate}
	\item If $k \leq n$, then $\langle M \rangle$ has dimension $k$, and minimum rank distance $n-k+1$.
	\item If $k<n$, then $\dim(\langle M \rangle \cap \langle M \rangle^{(q)})= k-1$ and $\dim(\langle M\rangle + \langle M\rangle^{(q)} )= k+1.$
	\item If $A \in \F_{q^m}^{k\times N}$ is another Moore matrix then $M+A$ is also a Moore matrix. Moreover, if the column rank of $A$ is equal to $r<n-k+1$, then the minimum rank distance of $\langle M+A\rangle$ is at least $n-k+1-r$.
	\item 
	If the minimum rank distance of $\langle M\rangle$ is $d>1$, then the minimum rank distance of $\langle M\rangle + \langle M^{([1])}\rangle $ is equal to $d-1$.
	\item 
	If the minimum rank distance of $\langle M\rangle$ is $d>1$, and $E \in \F_q^{N \times (N-s)}$ is a full rank matrix, then $ME$ is a Moore matrix and the minimum rank distance of $\langle ME\rangle = \langle M\rangle E$ is at least $d-s$.
\end{enumerate}
	
\end{lemma}
\begin{proof}
 \begin{enumerate}
	\item The case $n=N$ is given in Theorems 6 and 7 of \cite{Gabidulin85}. Thus we know that, if $N> n$, we can puncture the vector space $\langle M \rangle $ to get a space $\langle M'\rangle$ of length $n$, dimension $k$ and minimum rank distance $n-k+1$. Hence the minimum rank distance of $\langle M\rangle$ is at least $n-k+1$. That it cannot be greater follows from Lemma 4.7 of \cite{HTM15}.
	\item The first statement follows easily from the Moore matrix structure. This implies that \[\dim (\langle M\rangle + \langle M^{(q)} \rangle) = \dim (\langle M\rangle) +\dim (\langle M^{(q)}\rangle ) - \dim(\langle M \rangle \cap \langle M \rangle^{(q)} ) = k+1 .\]
	\item The first statement follows from the fact that $(x+y)^{[i]}= x^{[i]}+y^{[i]}$ for any $x,y \in \F_{q^m}$. Therefore the Moore structure is preserved under addition of matrices. For the second part note that any element $\boldsymbol{a}\in\langle A\rangle$ has rank at most $r$ and any non-zero element $\boldsymbol{m}_i\in\langle M\rangle$ has rank at least $n-k+1$. Hence $\boldsymbol{a}$ can change the rank of $\boldsymbol{m}\pm \boldsymbol{a}$ by at most $r$, i.e.\ the rank of any non-zero element of $\langle M+A \rangle$ has rank at least $n-k+1-r$.
	\item
	Since the minimum rank distance of $\langle M \rangle$ is $d>1$, it follows from 2.\ that $\dim(\langle M\rangle + \langle M^{([1])}\rangle) = k+1$. Then part 1.\ implies that 
	 the minimum rank distance of $\langle M\rangle + \langle M^{([1])}\rangle$ is $n-(k+1)+1=d-1$. 
	\item
	Let $E' \in \F_q^{n\times s}$ be such that $[E\mid E']$ has full rank. Then, $[E\mid E']$ is an isometry, and so $\langle G[E\mid E'] \rangle$ has minimum rank distance $d$. Removing the last $s$ columns gives $\langle GE\rangle$, which can only decrease the rank by at most $s$.

\end{enumerate}

\end{proof}



A well-known class of codes in the rank metric are the Gabidulin codes \cite{Gabidulin85}. Gabidulin codes are those whose generator matrix is a Moore matrix in which the generating vector has full rank:

\begin{definition}
	Fix $k \leq n \leq m$, and let $\boldsymbol{\alpha} = (\alpha_1,  \dots , \alpha_n) \in \F_{q^m}^n$, $\rk(\boldsymbol{\alpha}) = n$. The \textit{Gabidulin code} of length $n$ and dimension $k$ over $\F_{q^m}$, denoted by $\Gab_{n,k}(\alpha)$ is given by the row space of the matrix, \begin{equation}\label{eq:genmatrix}
	G = \left(\begin{array}{cccc} \alpha_1 & \alpha_2 & \ldots & \alpha_n \\ \alpha_1^{[1]} & \alpha_2^{[1]} & \ldots & \alpha_n^{[1]} \\ & & \vdots & \\ \alpha_1^{[k-1]} & \alpha_2^{[k-1]} & \ldots & \alpha_n^{[k-1]}. \end{array}\right).
\end{equation}
\end{definition}

From Lemma \ref{l:moore}, we have that Gabidulin codes are MRD codes. Moreover, they have efficient decoding algorithms \cite{Silva09,Wachter13,Gabidulin85}. Gabidulin codes are also closed under the linear isometries of the rank-metric (for isometries of rank-metric codes see e.g. \cite{Berger03,Morrison13}). Specifically, if $\beta \in \F_{q^m}^*$ and $\sigma \in\GL_n(\F_q)$, then $\beta \textnormal{Gab}_{n,k}(\boldsymbol{\alpha})\sigma = \textnormal{Gab}_{n,k}(\beta\boldsymbol{\alpha}\sigma)$.

\subsection{Decoding From an Arbitrary Generator Matrix}\label{ss:decoding}

McEliece cryptosystems based on Generalized Reed-Solomon (GRS) codes were effectively broken by Sidel'nikov-Shestakov \cite{Sidelnikov92}. Their attack allows one to recover the generating vector of a GRS code, and therefore a decoding algorithm. Similarly, in the case of Gabidulin codes, a decoding algorithm can be found if one knows the canonical generator matrix of the code. Using a simple method, we can also recover a decoding algorithm if the generator matrix is not in canonical form, as described in the following.

Consider the Gabidulin code $\textnormal{Gab}_{n,k}(\boldsymbol{\alpha})$ with dimension $1 < k < n$ and generator matrix $SG$, where $S\in\GL_k(\F_{q^m})$ and $G$ of the form \eqref{eq:genmatrix}. Then $\textnormal{Gab}_{n,k}(\boldsymbol{\alpha})^{([1])}\cap \textnormal{Gab}_{n,k}(\boldsymbol{\alpha})$ is the Gabidulin code $\textnormal{Gab}_{n,k-1}(\boldsymbol{\alpha}^{([1])})$ (see Lemma \ref{l:moore}). Iterating with this new Gabidulin code, we can eventually obtain a code of dimension $1$, which is generated by $\boldsymbol{\alpha}^{([k-1])}$. If we take some non-zero element of this space, it has the form $\beta \boldsymbol{\alpha}^{([k-1])}$, for some $\beta \in \F_{q^m}$. Applying the Frobenius map coordinate-wise $m-k+1$ times, we obtain an element of the form $\beta^{[m-k+1]}\boldsymbol{\alpha}$. Using this element, we can construct a generator matrix, $BG$, for $\textnormal{Gab}_{n,k}(\boldsymbol{\alpha})$ which will have the form $$BG = \left(\begin{array}{cccc} \beta^{[m-k+1]} &  &  &  \\ & \beta^{[m-k+2]} &  &  \\ & & \ddots & \\  &  &  & \beta \end{array}\right) \left(\begin{array}{cccc} \alpha_1 & \alpha_2 & \ldots &\alpha_n \\ \alpha_1^{[1]} & \alpha_2^{[1]} & \ldots & \alpha_n^{[1]} \\ & & \vdots & \\ \alpha_1^{[k-1]} & \alpha_2^{[k-1]} & \ldots & \alpha_n^{[k-1]}\end{array}\right).$$ The change of basis from $SG$ to $BG$ is then given by $BS^{-1}$. For a message $\boldsymbol{m} \in \F_{q^m}^k$, encoded as $\boldsymbol{m}SG$, we can now decode with respect to $\textnormal{Gab}_{n,k}(\beta^{[m-k+1]}\boldsymbol{\alpha})$ to obtain $\boldsymbol{m}SB^{-1}$. Then, applying $BS^{-1}$, we can recover $\boldsymbol{m}$.

\subsection{GPT and GGPT Cryptosystems}\label{ss:GPT}

Let $S\in \GL_k(\F_{q^m})$, $G \in \F_{q^m}^{k\times n}$ be a generator matrix of a Gabidulin code, say $\textnormal{Gab}_{n,k}(\boldsymbol{\alpha})$, capable of correcting $t'$ errors, and $X\in \F_{q^m}^{k\times n}$ be a matrix of column rank $t < t'$. We define 
\begin{equation}G_{\textnormal{pub}} := SG + X.  
\end{equation}
We call a \textit{GPT cryptosystem} one in which the public key is given by the pair 
\begin{equation}\label{eq:kpub}\kappa_{\textnormal{pub}} = (G_{\textnormal{pub}}, t'-t),\end{equation} and the private key is given by \begin{equation}\label{eq:ksec}\kappa_{\textnormal{pvt}} = (G, S).\end{equation} 
An encryption of a message $\boldsymbol{m} \in \F_{q^m}^k$ is given by 
$$\boldsymbol{m}G_{\textnormal{pub}} + \boldsymbol{e} = \boldsymbol{m}SG +\boldsymbol{m}X + \boldsymbol{e},$$ 
where $\boldsymbol{e}\in\F_{q^m}^n$ is a randomly chosen vector of rank at most $t'-t$. The product $\boldsymbol{m}S$ can be recovered from a decoding algorithm for $\textnormal{Gab}_{n,k}(\boldsymbol{\alpha})$ because all elements of $\langle X\rangle$ have weight at most $t$. Specifically, if $\wt_{\textnormal{R}}(\boldsymbol{e}) \leq t'-t$, 
$$\wt_{\textnormal{R}}(\boldsymbol{m}X + \boldsymbol{e}) \leq \wt_{\textnormal{R}}(\boldsymbol{m}X) + \wt_{\textnormal{R}}(\boldsymbol{e}) \leq t' .$$ Inverting $S$, the message $\boldsymbol{m}$ can then be recovered. 
We will call the elements of the form $\boldsymbol{m}X$ the \textit{designed error} associated with the encryption of $\boldsymbol{m}$, and $X$ the \textit{designed error matrix}.

\vspace{0.3cm}

In \cite{Rashwan10, Loidreau10} the authors consider an alternative version which we call the \emph{generalized GPT (GGPT) cryptosystem}. This system uses a public matrix of the form 
\begin{equation}\label{eq:kpubmod} \hat{G}_{\textnormal{pub}} := S[X \mid G]\sigma \in \F_{q^m}^{k\times (n+\hat t)},\end{equation} 
where $G$ is as before, $X \in \F_{q^m}^{k\times \hat t}$ is a matrix of column rank $\hat t$, $S \in \GL_k(\F_{q^m})$, and $\sigma \in \GL_{n+\hat t}(\F_q)$. The public key is given by
\begin{equation}\kappa_{\textnormal{pub}} = (\hat{G}_{\textnormal{pub}}, t'),\end{equation} and the private key is given by \begin{equation}\kappa_{\textnormal{pvt}} = (G, S, \sigma).\end{equation} 
In the GGPT cryptosystem, an encryption of $\boldsymbol{m} \in \F_{q^m}^k$ is given by 
$$\boldsymbol{m}\hat{G}_{\textnormal{pub}} + \boldsymbol{e},$$ with $\rk(\boldsymbol{e}) \leq t'$. To recover $\boldsymbol{m}$, one first computes $$(\boldsymbol{m}\hat{G}_{\textnormal{pub}} + \boldsymbol{e})\sigma^{-1},$$ and then ignores the first $\hat t$ coordinates. Decoding the last $n$ coordinates with respect to $\textnormal{Gab}_{n,k}(\boldsymbol{\alpha})$, one obtains $\boldsymbol{m}S$, and by applying $S^{-1}$, the message $\boldsymbol{m}$ can be recovered.

\subsection{Overbeck's Attack}\label{ss:Overbeck}

We will describe Overbeck's attack from \cite{Overbeck08} for the case of the GGPT cryptosystem; the attack for the GPT case is analogous. Let $G \in \F_{q^m}^{k\times n}$ be a generator matrix for the Gabidulin code, $\textnormal{Gab}_{n,k}(\boldsymbol{\alpha})$. The first step in Overbeck's attack is to consider the extended matrix (for some $u\geq 1$) 
$$G_{\textnormal{ext}} :=\left(\begin{array}{c} S[X\mid G]\sigma \\ (S[X\mid G]\sigma)^{([1])} \\ \vdots \\ (S[X\mid G]\sigma)^{([u])}\end{array}\right) = \tilde{S}\left(\begin{array}{c|c} X & G \\ X^{([1])} & G^{([1])} \\ \vdots & \vdots \\ X^{([u])} & G^{([u])} \end{array}\right) \sigma.$$ 
Since the $n$ right-most columns of $G_{\textnormal{ext}}\sigma^{-1}$ span the Gabidulin code $\textnormal{Gab}_{k+u}(\boldsymbol{\alpha})$, the matrix can be brought into the form of 
\begin{equation}\label{eq:GGPTform} G_{\textnormal{ext}}' = \tilde{S}'\left(\begin{array}{c|c} X^* & G^* \\ X^{**} & 0\end{array}\right) \sigma,\end{equation} 
by some suitable row transformation, where $X^* \in \F_{q^m}^{(k+u)\times \hat t}$, $G^* \in \F_{q^m}^{(k+u)\times n}$ is a generator matrix of $\textnormal{Gab}_{k+u}(\boldsymbol{\alpha})$, and $X^{**}\in\F_{q^m}^{(k-1)u\times \hat t}$. If $X^{**}$ has rank $\hat t$, then any element of $\langle G_{\textnormal{ext}}'\rangle^{\perp}=\langle G_{\textnormal{ext}}\rangle^{\perp}$ has the form $\sigma^{-1}[0\mid \boldsymbol{h}]$, where $\boldsymbol{h} \in \textnormal{Gab}_{k+u}(\boldsymbol{\alpha})^\perp$. With this information one can reconstruct the code $\textnormal{Gab}_{n,k}(\boldsymbol{\alpha})$ and recover the encrypted message.

In the case when $X^{**}$ does not have full rank, Overbeck's attack fails, since $\textnormal{Gab}_{n,k}(\boldsymbol{\alpha})$ cannot be reconstructed from the dual of $\langle G_{\textnormal{ext}}\rangle^{\perp}$.
This is why, in \cite{Loidreau10}, Loidreau suggests to use a randomly chosen $X$ of low rank, $a$, since then the rank of $X^{**}$ can be bounded above. Specifically, to resist Overbeck's attack, one should choose $\hat t > (n-k)a$. However, this would drastically increase the key size of the cryptosystem. 
To avoid this problem of large key size, the Smart Approach considered in \cite{Rashwan10}, is to design $X$ in a structured way so that $X^{**}$ is rank-deficient, without necessarily having to increase $\hat t$. However the structure of $X$ makes the Smart Approach more vulnerable to attacks. 
A more detailed description of these two systems is given in Section \ref{s:ext}.

\section{Preliminary Results}\label{s:presults}

In this section, we show that one can decompose a matrix (or vector) of low column rank into the product of two matrices, one of which has full column rank, and the other with elements restricted to $\F_q$. 
Moreover, we prove some results about the coordinate-wise Frobenius map, as well as the structure of the designed error matrix, which we will need later on in our attack.

	We will denote by $M_{t\times n, r}(\F_q)$ the set of $t\times n$ matrices over $\F_q$ with rank $r$.
	The sphere around the origin of rank radius $t$ in $\F_{q^m}^n$ will be denoted by $$S_{n,t}^{\textnormal{R}}(\F_{q^m}) := \{\boldsymbol{x} \in \F_{q^m}^n \mid \rk(\boldsymbol{x}) = t\}.$$

\begin{proposition}\label{p:vbundle} 
	$$S_{n,t}^{\textnormal{R}}(\F_{q^m}) \cong S_{t,t}^{\textnormal{R}}(\F_{q^m}) \times M_{t\times n,t}(\F_q) / \GL_t(\F_q),$$ where $M_{t\times n,t}(\F_q) / \GL_t(\F_q)$ is the set of equivalence classes of $t\times n$ matrices over $\F_q$ of rank $t$, where two matrices are equivalent if they have the same row span.
\end{proposition}

\begin{proof}
As representatives of the cosets in $M_{t\times n}(\F_q) / \GL_t(\F_q)$ we consider the reduced row echelon form of the respective row span of the elements of the coset. 
	Define the map 
\begin{align*}
 \varphi \colon S_{t,t}^{\textnormal{R}}(\F_{q^m}) \times M_{t\times n}(\F_q) / \GL_t(\F_q) &\longrightarrow S_{n,t}^{\textnormal{R}}(\F_{q^m})\\
(\boldsymbol{v},U)&\longmapsto \boldsymbol{v}U . 
\end{align*}
We now show that $\varphi$ is bijective.

We first show that $\varphi$ is surjective. For this consider an arbitrary element in the image of $\varphi$, i.e.\ a vector $\boldsymbol{x} \in \F_{q^m}^n$ of rank $t$, and let $x_{i_1},  \dots , x_{i_t}$ be the first $t$ independent entries of $\boldsymbol{x}$, in positions $i_1,  \dots , i_t$. Then, the remaining $n-t$ entries of $\boldsymbol{x}$ can be expressed as an $\F_q$-linear combination of $x_{i_1}, \dots ,x_{i_t}$, thus we can write $\boldsymbol{x} = (x_{i_1},  \dots , x_{i_t}) M$ for some matrix $M  \in \F_q^{t\times n}$. Then there exists $S\in \GL_t(\F_q)$ such that $U=S-M$ is in reduced row echelon form. We get  $ (x_{i_1},  \dots , x_{i_t})S\in S_{t,t}^{\textnormal{R}}(\F_{q^m})$ and  $\boldsymbol{x} = \varphi( (x_{i_1},  \dots , x_{i_t})S,U)$, thus $\varphi$ is surjective. 

To show injectivity, suppose that there are two preimages, i.e.\ $\boldsymbol{x} = \varphi(\boldsymbol{v}, U) = \varphi(\boldsymbol{v'}, U') $. Without loss of generality, we can assume that $U= [ I_t \mid *]$. Denote by $U'_j$ the $j$th column of $U'$. Then we have
\begin{align*}
 (x_1,\dots,x_t) = \boldsymbol{v} = (\boldsymbol{v'} U'_1,\dots,\boldsymbol{v'} U'_t ) .
\end{align*}
Since $\boldsymbol{v}$ has rank $t$, $U'_1,\dots, U'_t$ must be non-zero. Because $U'$ is in reduced row echelon form, we get $U'=[I_t\mid *]$ and hence
$$(x_1,\dots,x_t) = \boldsymbol{v}= \boldsymbol{v'}.$$
We furthermore have $x_j = \boldsymbol{v}U_j = \boldsymbol{v'}U'_j$ for $j=t+1,\dots, n$. Thus
$$  \boldsymbol{v}U_j = \boldsymbol{v'}U'_j \iff  \boldsymbol{v}U_j = \boldsymbol{v}U'_j \iff  \boldsymbol{v}(U_j-U'_j)=0 .$$
Since $\rk(\boldsymbol{v})=t$, we get  $U_j-U'_j=0$ for $j=t+1,\dots,n$. Thus $U=U'$ and we have shown that $\varphi$ is injective.
\end{proof}

One can think of the space $M_{t\times n,t}(\F_q) / \GL_t(\F_q)$ as a set of matrices parameterizing the Grassmannian $\Gr(t, \F_q^n)$, i.e.\ the space of $t$-dimensional subspaces of $\F_q^n$.  According to the proof of Proposition \ref{p:vbundle}, we can express a vector $\boldsymbol{x}$ of rank $t$ as $\boldsymbol{x} = \hat{\boldsymbol{x}}U$ for any matrix representation $U$ of a certain element of the Grassmannian $\Gr(t, \F_q^n)$.

We can easily extend the result of Proposition \ref{p:vbundle} from vectors of rank $t$ to matrices of column rank $t$. Then we get the following result.
\begin{corollary}\label{c:Gsupport}
 Let $X \in \F_{q^m}^{k\times n}$ be a matrix of rank $k$ and column rank $t$. Then there exist $V\in \F_{q^m}^{k\times t}$ with $\rk(V)=k$ and $ U\in \F_q^{t\times n}$ with $\rk(U)=t$, such that $$X=VU.$$
\end{corollary}

\begin{definition}
 Let $X \in \F_{q^m}^{k\times n}$ be a matrix of rank $k$ and column rank $t$ and $V\in \F_{q^m}^{k\times t}, U\in \F_q^{t\times n}$ such that $X=VU$. We call $\langle U \rangle$ the  \textit{Grassmann support} of $X$ which will be denoted by $\langle U\rangle = \supp_{\Gr}(X)$. By abuse of notation we will also call any matrix representation $U\in \F_q^{t\times n}$ of this space the Grassmann support of $X$.
\end{definition}

\begin{lemma}\label{lem:incl}
 	Let $X\in \F_{q^m}^{k\times n}$ be a matrix of rank $k$ and column rank $t\geq k$. Then $\langle X\rangle \subseteq \supp_\Gr (X) $ and the inclusion is strict if and only if $t>k$.
\end{lemma}
\begin{proof}
Let $U \in \F_q^{t\times n}$  be the  Grassmann support of $X$.
 By Corollary \ref{c:Gsupport}, we can write $X=VU$ for some $V\in \F_{q^m}^{k\times t}$. Thus every row of $X$ is a $\F_{q^m}$-linear combination of the rows of $U$, which implies that $\langle X\rangle \subseteq \langle U\rangle$. Since $\dim(\langle U\rangle)=t$ and $\dim(\langle X\rangle)=k$, we get equality if and only if $k=t$.
\end{proof}

%
%
%

The following two lemmas are needed to prove the main results of this section in Theorems \ref{t:frobchain} and \ref{t:fullsum}.
\begin{lemma}\label{l:crk}
Let $X \in \F_{q^m}^{k\times n}$ be a matrix of column rank $t$ and $S \in \GL_k(\F_{q^m})$. Then, $SX$ also has column rank $t$.
\end{lemma}
\begin{proof}
Denote the $i$th column of $X$ by $X_{i}$. 
Assume that $SX$ has column rank less than $t$, i.e.\ for any $i_{1}<\dots <i_{t}\in \{1,\dots,n\}$ there exist $a_{1},\dots, a_{t}\in \F_{q}$ such that
\[\sum_{\ell=1}^{t} a_{\ell} (SX)_{i_{\ell}}  = 0 \iff S\sum_{\ell=1}^{t} a_{\ell} X_{i_{\ell}} =0 \iff \sum_{\ell=1}^{t} a_{\ell} X_{i_{\ell}} =0 .\]
This is a contradiction to the fact that the column rank of $X$ is $t$.
\end{proof}

The following properties of the coordinate-wise Frobenius map will be used throughout the paper. The first statement follows straightforwardly from the $\F_q$-linearity of the Frobenius map, the second and the third are known and can be found, for instance, in \cite{Giorgetti10,HTM15}.

\begin{lemma}\label{l:results} 
The following hold for any prime power $q$ and $0 < n \leq m$.
	\begin{enumerate}
		\item Let $\boldsymbol{x} \in \F_{q^m}^n$ have rank $r$. Then, $\boldsymbol{x}^{(q)}$ also has rank $r$.
		\item Let $M \in \GL_n(\F_{q^m})$. Then, $(M^{-1})^{(q)} = (M^{(q)})^{-1}$.
		\item Let $\mathcal{S}\subset \F_{q^m}^{n}$ be an $\F_{q^m}$-subspace. Then, $\mathcal{S}^{(q)} = \mathcal{S}$ if and only if $\mathcal{S}$ has a basis contained in $\F_q^{n}$. 
	\end{enumerate}
\end{lemma}

We saw in Lemma \ref{lem:incl} that if a matrix $X$, with Grassmann support $U$, has column rank which is greater than its rank, then $\langle X\rangle \subsetneq \langle U\rangle$. The following theorem shows that we can use the Frobenius map to recover $\langle U\rangle$ from $X$.

\begin{theorem}\label{t:frobchain}
	Let $X\in \F_{q^m}^{k\times n}$ be a matrix of column rank $s$. Then, for any $\ell\geq 0$,
	$$\sum_{i=0}^{s-1}\langle X\rangle^{([i])} = \sum_{i=0}^{s+\ell}\langle X\rangle^{([i])} = \supp_\Gr (X).$$ 
In particular, $$\dim\left( \sum_{i=0}^{s-1} \langle X\rangle^{([i])}\right) = s.$$ 
	\end{theorem}

\begin{proof}
	The chain of subspaces 
	$$\langle X \rangle \subseteq \langle X\rangle + \langle X\rangle^{(q)}  \subseteq  \sum_{i=0}^{2} \langle X \rangle^{([i])}  \subseteq  \dots$$ must eventually stabilize. Let $\ell$ be such that, 
	$$\sum_{i=0}^{\ell-1} \langle X\rangle^{([i])} = \sum_{i=0}^{\ell} \langle X\rangle^{([i])}.$$ 
Define $s':=\dim\sum_{i=0}^{\ell-1} \langle X\rangle^{([i])} $. We have 
	$$\left(\sum_{i=0}^{\ell-1} \langle X\rangle^{([i])}\right)^{(q)} = \sum_{i=1}^{\ell} \langle X\rangle^{([i])} \subseteq \sum_{i=0}^{\ell}\langle X\rangle^{([i])} =  \langle X\rangle^{([\ell])} + \sum_{i=0}^{\ell-1} \langle X\rangle^{([i])}.$$ 
By Lemma \ref{l:moore}, $\dim\sum_{i=0}^{\ell-1} \langle X\rangle^{([i])} = \left(\dim\sum_{i=0}^{\ell-1} \langle X\rangle^{([i])}\right)^{(q)}  =s'$. Hence, we must have
	$$\left(\sum_{i=0}^{\ell-1} \langle X\rangle^{([i])}\right)^{(q)} =
\sum_{i=0}^{\ell-1} \langle X\rangle^{([i])},$$ 
	and therefore we can use the third point of Lemma \ref{l:results} and express the sum on the right as the row space of a matrix $U' \in \F_q^{s'\times n}$ of (column) rank $s'$. Thus there exists $S\in \GL_{k\ell}(\F_{q^m})$ such that 
$$ \left(\begin{array}{c} X \\ X^{([1])} \\ \vdots \\ X^{([\ell-1])}\end{array}\right) = S \left(\begin{array}{c} U'\\ 0 \end{array} \right).
$$
This implies that $\langle X\rangle \subseteq \langle U' \rangle$. It follows from Proposition \ref{p:vbundle} that $s'\geq s$. Moreover, by Lemma \ref{l:crk}, the above matrix on the left has column rank $s'$. Since, by the $\F_q$-linearity of the Frobenius, the column rank of this matrix is equal to the column rank of $X$ we get $s=s'$ and hence $ \supp_\Gr (X) = \langle U'\rangle$.
\end{proof}

Note that a matrix $X \in \F_{q^m}^n$ can always be decomposed into a Moore matrix component $X_{\textnormal{Moore}}$ and a non-Moore matrix component $Z$ as
\[ X = X_{\textnormal{Moore}} +Z   .\] 
\begin{definition}
 We will call such a decomposition a \emph{Moore decomposition}. There exists a Moore decomposition so that the non-Moore component has lowest possible column rank. In this case, we call the Moore decomposition a \emph{minimum column rank Moore decomposition}. 
\end{definition}

Proposition \ref{p:uniquesupport} shows that, regardless of the choice of Moore decomposition, the Grassmann support of a non-Moore matrix component of a minimum column rank  Moore decomposition is the same.

\begin{proposition}\label{p:uniquesupport} 
Suppose that $X\in \F_{q^{m}}^{k\times n}$ is a matrix which has minimum column rank Moore decomposition $X = A_{\textnormal{Moore}} + A$, where $A_{\textnormal{Moore}}$ is a Moore matrix, and $A$ has column rank $s$. Then, any other minimum column rank Moore decomposition $X = B_{\textnormal{Moore}} + B$ satisfies that $ \supp_\Gr (A) =  \supp_\Gr (B)$.
\end{proposition}

\begin{proof}
	Let $A$ have Grassmann support $U$, and $B$ have Grassmann support $V$. I.e., we can write $A = A'U$ and $B = B'V$ with $U,V\in\F_q^{s\times n}$ of full rank. Let $E \in \F_q^{(n-s)\times n}$ be a parity check matrix for $\langle V \rangle$. Then, 
	$$B_{\textnormal{Moore}}E^T = XE^T -BE^T = XE^T = A_{\textnormal{Moore}}E^T + AE^T,$$ 
	which yields $$(B_{\textnormal{Moore}}-A_{\textnormal{Moore}})E^T = AE^T.$$ 
	Since $E$ is a matrix over $\F_{q}$, $(B_{\textnormal{Moore}} - A_{\textnormal{Moore}})E^T$ is a Moore matrix, therefore the matrix $AE^T$
must be a Moore matrix as well. This gives that 
$(AE^T)_i = (A_1E^T)^{([i-1])}=A_1^{([i-1])}E^T$
 for $i = 2,  \dots , k$. Since $A$ itself is not necessarily a Moore matrix, row $i$ of $A$ must be of the form $A_i \in A_{1}^{([i-1])} + \ker(E),$ for $i = 1,  \dots , k$. Then, we can write 
$$A = \underbrace{\left(\begin{array}{c} A_1 \\ A_1^{([1])} \\ \vdots \\ A_1^{([k-1])} \end{array}\right)}_{\bar{A}} + \underbrace{\left(\begin{array}{c} \kappa_1 \\ \kappa_2 \\ \vdots \\ \kappa_k \end{array}\right)}_{\kappa'V},$$ 
for $\kappa_1,\dots,\kappa_k \in \ker(E) = \langle V\rangle$ and $\kappa' \in \F_{q^m}^{k\times s}$. If we let $F \in \F_q^{(n-s)\times n}$ be a parity check matrix for $\langle U\rangle$, then $AF^{T}=A'UF^T=0$ and hence in particular $A_{1}F^{T}=0$. Since $F$ is a matrix over $\F_{q}$, we also get $A_{1}^{([i])}F^{T}=0$ for $i=1,\dots,k-1$. Hence,
$$0 = AF^T = \bar{A}F^T + \kappa'VF^T = \kappa' VF^T.$$ 
Since $X=(A_{\textnormal{Moore}} + \bar{A}) + \kappa'V$ is also a Moore decomposition of $X$, then the column rank of $\kappa' V$ must be equal to $s$ and so $\langle V\rangle = \langle \kappa'V\rangle$. Thus, 
$$\langle V\rangle F^T = 0,$$ and therefore, $\langle V\rangle = \langle U \rangle$, so the Grassmann supports of $A$ and $B$ are the same.
\end{proof}

\begin{theorem}\label{t:fullsum}
	Let $M \in \F_{q^m}^{k\times n}$ be a Moore matrix and $X\in \F_{q^m}^{k\times n}$ be of column rank $s$, where $s$ is the rank of the non-Moore component in a minimum column rank Moore decomposition of $X$. Then, we have $$\sum_{i=0}^{s}\langle M+X\rangle^{([i])} = \sum_{i=0}^{s}\langle M\rangle^{([i])} + \supp_\Gr (X).$$
\end{theorem}

\begin{proof}
Let $U \in \F_q^{s\times n}$ be the Grassmann support of $X$.  
Moreover, let $X_{i}, M_{i}$ denote the $i$th row of $X$ and $M$ respectively, and let
$$X' = \left(\begin{array}{c} X_1^{([1])}-X_2 \\ X_2^{([1])}-X_3 \\ \vdots \\ X_{k-1}^{([1])} -X_{k}\end{array}\right), 
M^{*} = \left(\begin{array}{c} M_{1}\\M_1^{([1])} \\  \vdots \\ M_{1}^{([k+s-1])} \end{array}\right),
X^{*} = \left(\begin{array}{c} X_{1}\\X_1^{([1])} \\ \vdots \\ X_{1}^{([s-1])}\\\hline  X_{1}^{([s])}\\ X_{2}^{([s])} \\ \vdots\\ X_{k}^{([s])}\end{array}\right).$$
	Then the space 
	$\sum_{i=0}^{s}\langle M+X\rangle^{([i])}$
	is generated by the row span of 
	\begin{equation}\label{eq:expmatrix} \left(\begin{array}{c} M+X \\ (M+X)^{([1])} \\ \vdots \\ (M+X)^{([s])}\end{array}\right) = \tilde{S}\left(\begin{array}{c} M^* + X^* \\ X' \\ \vdots \\ (X')^{([s-1])}\end{array}\right),\end{equation} 
	for a suitable row transformation matrix $\tilde{S}$.
Since $U\in \F_q^{s\times n}$ we have $U^{([i])} = U$ for $i\geq 0$. It follows that the rows of  $X'$ are  elements of $\langle U\rangle$, which implies that the Grassmann support $\langle U'\rangle$ of $X'$ is a subspace of $\langle U\rangle$ and hence that $X'$ has column rank $s' \leq s$. By Theorem \ref{t:frobchain}, $$\sum_{i=0}^{s'-1} \langle X'\rangle^{([i])} = \sum_{i=0}^{s-1} \langle X'\rangle^{([i])} = \langle U' \rangle \subseteq \langle U\rangle .$$ 

We now want to show that $\langle U' \rangle = \langle U\rangle $. Suppose for the sake of contradiction that the rank of $U'$ is strictly smaller than $s$. We write $X$ as a Moore decomposition 
$$X = \left(\begin{array}{c} X_1 \\ X_2 \\ \vdots \\ X_{k}\end{array}\right) = \left(\begin{array}{c} X_1 \\ X_1^{([1])} \\ \vdots \\ X_1^{([k-1])}\end{array}\right) + \underbrace{\left(\begin{array}{c} 0 \\ X_2-X_1^{([1])} \\ \vdots \\ X_{k}-X_1^{([k-1])}\end{array}\right)}_{X''}.$$ 
We note that $X_{i+1}-X_i^{[1]} \in \langle U'\rangle$ for $i= 1, ...,k-1$. Starting from the first non-zero row of $X''$, it follows that
$$(X_2 - X_1^{([1])})^{([1])} = X_2^{([1])} - X_1^{([2])} \in \langle U'\rangle $$
which implies 
\begin{align*} & X_2^{([1])} - X_1^{([2])} - (X_2^{([1])} - X_3) \in \langle U'\rangle \\ 
\iff \quad & X_3 - X_1^{([2])} \in \langle U'\rangle.\end{align*} We recognize this as the second non-zero row of $X''$. Continuing in this fashion, we can obtain that every row of $X''$ must belong to $U'$. Hence, $X''$ has column rank  at most $ s'<s$. 
However, this contradicts 
the fact that the minimal column rank Moore decomposition has non-Moore part with column rank $s$. Therefore,  by Proposition \ref{p:uniquesupport}, $U'$ has rank $s$ and we have $\langle U'\rangle = \langle U\rangle$. 

Hence, we have shown that the row space of  the second matrix in \eqref{eq:expmatrix} is equal to the row space of
$$\left(\begin{array}{c} M^* + X^* \\ U\end{array}\right)$$
which is in turn equal to the row space of
$$ \left(\begin{array}{c} M^* \\ U\end{array}\right),$$ 
because we can cancel $X^*$ by taking suitable  elements of $\langle U\rangle$, since $\langle X^{*}\rangle\subseteq \langle U'\rangle = \langle U\rangle$. This implies the statement.
\end{proof} 

\begin{lemma}\label{l:supportcontain}
	Let $X \in \F_{q^m}^{k\times n}$ have minimum column rank Moore decomposition, $X = X_{\textnormal{Moore}} + Z.$ Then, $$  \supp_{\textnormal{Gr}}(X_{\textnormal{Moore}}) + \supp_{\textnormal{Gr}}(Z) = \supp_{\textnormal{Gr}}(X).$$
	In particular, $\crk(X_{\textnormal{Moore}}) \leq \crk(X)$.
\end{lemma}

\begin{proof}
Define $\ell:=\max(\crk(X), \crk(X_{\textnormal{Moore}}))$. 
	Using Theorems \ref{t:frobchain} and \ref{t:fullsum}, we have 
	\begin{align*} \supp_{\textnormal{Gr}}(X) &= \sum_{i=0}^{\ell} \langle X\rangle^{([i])} \\ 
&= \sum_{i=0}^{\ell}\langle X_{\textnormal{Moore}} + Z\rangle^{([i])} \\ 
&= \sum_{i=0}^{\ell}\langle X_{\textnormal{Moore}}\rangle^{([i])} + \supp_{\textnormal{Gr}}(Z)\\
&=  \supp_{\textnormal{Gr}}(X_{\textnormal{Moore}}) + \supp_{\textnormal{Gr}}(Z) . 
\end{align*} 
\end{proof}

\begin{corollary}\label{c:rankone}
	Let $M \in \F_{q^m}^{k\times N}$ be a Moore matrix and $X$ be of column rank $t$ with minimum column rank Moore decomposition $X = X_{\textnormal{Moore}} + Z$, where $\crk(Z) = s$. Suppose that $d_{\textnormal{min}}^{\textnormal{R}}(\langle M\rangle) \geq s+t+2$. Then, all elements of rank one in $$\sum_{i=0}^{s}\langle M+X\rangle^{([i])},$$ belong to $\supp_{\textnormal{Gr}}(X)$. Moreover, if $s=t$, the elements of rank one exactly span $\supp_{\textnormal{Gr}}(X) = \supp_{\textnormal{Gr}}(Z)$. 
\end{corollary}

\begin{proof}
	Let $\mathcal{U}$ be the subspace spanned by all elements of rank one in $$\sum_{i=0}^{s}\langle M+X\rangle^{([i])}.$$ From Lemma \ref{l:supportcontain}, if $X = X_{\textnormal{Moore}} + Z$, is a minimum column rank decomposition, then we know that $\supp_{\Gr}(Z) \subseteq \supp_{\Gr}(X).$ Let $H \in \F_q^{(n-t)\times n}$ be parity check matrix for $\supp_{\textnormal{Gr}}(X)$. From Lemma \ref{l:moore}, we have \begin{align*} d_{\textnormal{min}}^{\textnormal{R}} \left(\sum_{i=0}^{s}\langle M+X\rangle^{([i])}H^T\right) &= d_{\textnormal{min}}^{\textnormal{R}} \left(\sum_{i=0}^{s}\langle M\rangle^{([i])}H^T\right) \\ &\geq (s+t+2)-s-t \\ &= 2.\end{align*} 
	Since $H$ is a matrix over $\F_q$, we get $\wt_{\textnormal{R}}(\boldsymbol{x}) \leq \wt_{\textnormal{R}}(\boldsymbol{x}H)$, and therefore we must have that $\mathcal{U}\subseteq \supp_{\Gr}(X)$. By Theorem \ref{t:fullsum}, $\supp_{\Gr}(Z) \subseteq \mathcal{U}$ and if $s=t$ then $\supp_{\Gr}(Z) = \supp_{\Gr}(X)$. Therefore  we have $$\supp_{\Gr}(Z) = \mathcal{U} = \supp_{\Gr}(X).$$
\end{proof}

To set up our attack in Section \ref{s:attack}, we need to find the 
elements of rank one in a linear rank metric code efficiently. To accomplish this, we only need to find the codewords that have all coordinates in $\F_q$ (all other rank one codewords are multiples of these). 
The following lemma shows how these codewords in $\F_{q}^n$ can be computed.
\begin{lemma}\label{l:rankone}
Let $G\in \F_{q^m}^{k\times n}$ be in reduced row echelon form and denote by $G_i$  the $i$-th row of $G$. Then the solutions to 
\begin{equation}\label{eq10} 
\sum_{i=1}^{k} a_i (G_i^{([1])}-G_i) = 0,\end{equation} 
for variables $a_i \in \F_q$, represent the codewords of $\langle G \rangle$ in $\F_q^n$.
\end{lemma}
\begin{proof}
Any codeword can be written as an $\F_{q^{m}}$-linear combination of the rows of $G$. Since all rows of $G$ have their pivot equal to $1$, a codeword with entries only in $\F_{q}$ needs to be an $\F_{q}$-linear combination of the rows. Thus, we get that any codeword in $ \F_q^n$ can be written as
$\sum_{i=1}^{k}a_i G_i$
for some $a_{i}\in \F_{q}$. Furthermore, we know that 
$$\boldsymbol{v} \in \F_{q}^{n} \iff \boldsymbol{v}^{([1])} - \boldsymbol{v} = 0 ,$$
hence 
$$\sum_{i=1}^{k}a_i G_i \in \F_q^n \iff \sum_{i=1}^{k} a_i (G_i^{([1])}-G_i) = 0 .$$
\end{proof}

When expanded over $\F_q$, Equation (\ref{eq10}) gives rise to a linear system of equations with $k$ variables, which can efficiently be solved with standard methods.


\section{Cryptanalysis of the GPT Cryptosystem}\label{s:attack}

In this section we explain our new attack to break the GPT cryptosystem, as defined in Subsection \ref{ss:GPT}. Our attack extends Overbeck's attack to cryptanalyze the system for all parameters. In Section \ref{s:ext}, we show how this same idea can be used to cryptanalyze the GGPT variant. 

Recall that the public key generator matrix is of the form
$$G_{\textnormal{pub}} := SG+X \in \F_{q^m}^{k\times n},$$
where $G$ is a generator matrix of a Gabidulin code $\Gab_{n,k}(\boldsymbol{\alpha})$, $X \in \F_{q^m}^{k\times n}$ is a matrix of column rank $t$, and $S \in \GL_k(\F_{q^m})$.

Note that, as an attacker, we do not have a priori knowledge of the parameter $s$ (the column rank of the non-Moore part in the minimal column rank Moore decomposition of $X$). We can generally assume $s = t$, or else start with $s=1$ and increase the value up to $t$ until the attack succeeds. 

\begin{theorem}\label{t:break}
Consider a GPT cryptosystem as defined in Subsection \ref{ss:GPT},  with public key generator matrix $G_{\textnormal{pub}} = SG + X \in \F_{q^m}^{k\times n}$. Let $S^{-1}X = X_{\textnormal{Moore}} + Z$ be a minimal column rank Moore decomposition with $s=\mathrm{colrk}(Z)$. Suppose an adversary can find a full rank matrix $U \in \F_q^{s'\times n}$ for $s \leq s' \leq t$ satisfying 
$$\supp_{\Gr}(Z) \subseteq \langle U\rangle\subseteq \supp_{\Gr}(X),$$ 
then an encrypted message from a public key of the form \eqref{eq:kpub} can be recovered in polynomial time.
\end{theorem}

\begin{proof}
	Let $H \in \F_q^{(n-s') \times n}$ be a parity check matrix for $\langle U \rangle$. Applying $H$ to the public key generator matrix yields 
	$$G_{\textnormal{pub}}H^T = (SG+X)H^T = S(G+X_{\textnormal{Moore}})H^T .$$
From Lemma \ref{l:supportcontain} we know that  $\mathrm{colrk}(X_{\textnormal{Moore}}) \leq t$. Then, 
from Lemma \ref{l:moore}, it follows that $\langle G+X_{\textnormal{Moore}} \rangle$ has minimum rank distance at least $n-k+1-t$, and that $\langle G+X_{\textnormal{Moore}} H^T\rangle$ has minimum rank distance at least $n-k+1-(t+s')$. Moreover, $GH^T+X_{\textnormal{Moore}}H^T$ is a Moore matrix. 
	
	 From the minimum distance we know that there are $n-(t+s')$ independent columns in this matrix, which generate a Gabidulin code of minimum distance $n-(t+s')-k+1$, $\Gab_{n-(t+s'),k}(\boldsymbol{\gamma})$, for some $\boldsymbol{\gamma} \in \F_{q^m}^{n-(t+s')}$. From Subsection \ref{ss:decoding}, we can recover a decoding algorithm for $\Gab_{n-(t+s'),k}(\boldsymbol{\gamma})$ with respect to the submatrix formed by these $n-(t+s')$ columns. The error correction capability of $\Gab_{n-(t+s'),k}(\boldsymbol{\gamma})$ is 
	 $$\left\lfloor\dfrac{n-(t+s')-k}{2}\right\rfloor = \left\lfloor t' - \dfrac{t+s'}{2}\right\rfloor \geq t'-t \geq \rk(\boldsymbol{e})\geq \rk(\boldsymbol{e}H^T),$$ 
	 where the last inequality follows from the fact that $H$ is a matrix over $\F_q$. 
	For an encrypted message $\boldsymbol{m}(SG+X) + \boldsymbol{e}$, we have 
	 $$(\boldsymbol{m}(SG+X) + \boldsymbol{e})H^T = \boldsymbol{m}S(GH^T+X_{\textnormal{Moore}}H^T) + \boldsymbol{e}H^T.$$ 
	 When we restrict this to the above chosen independent columns, we can uniquely decode in the respective code $\Gab_{n-(t+s') ,k}(\boldsymbol{\gamma})$ and can therefore recover $\boldsymbol{m}$.
\end{proof}



We can now use the previous result to attack and break the GPT cryptosystem.

\begin{corollary}\label{c:break}
Consider a GPT cryptosystem as defined in Subsection \ref{ss:GPT} with public key generator matrix $G_{\textnormal{pub}} = SG + X \in \F_{q^m}^{k\times n}$. For any such cryptosystem, an encrypted message can be recovered in polynomial time.
\end{corollary}

\begin{proof}
	As before, let $S^{-1}X = X_{\textnormal{Moore}} + Z$ be a minimal column rank Moore decomposition. Denote by $s$ the column rank of $Z$. We first note that $d_{\textnormal{min}}^{\textnormal{R}}(\langle G\rangle) \geq s+t+2$, since
	$$\dfrac{d_{\textnormal{min}}^{\textnormal{R}}(\langle G\rangle) - 1}{2} \geq  \left\lfloor\dfrac{n-k}{2}\right\rfloor = t' > t \geq \dfrac{s+t}{2}.$$ 	
	By Corollary \ref{c:rankone}, all the elements of rank one in $\sum_{i=0}^{s}\langle G+X\rangle^{([i])}$ belong to the Grassmann support of $X$.
	With Lemma \ref{l:rankone} we can find a basis matrix $U\in \F_q^{s'\times n}$ for these elements of rank one in polynomial time. 
	We have $\langle U\rangle \subseteq \supp_\Gr(X)$.  
On the other hand, by Theorem \ref{t:fullsum}, $\supp_{\Gr}(Z) \subseteq \sum_{i=0}^{s}\langle G+X\rangle^{([i])}$. Thus, we also have $\supp_{\Gr}(Z) \subseteq \langle U\rangle$. 
Therefore we can use Theorem \ref{t:break} to recover the encrypted message.
\end{proof}

\section{Cryptanalysis of GGPT Variants}\label{s:ext}

In this section we adapt our attack to break the GGPT cryptosystem, as defined in Subsection \ref{ss:GPT}. To do so we will consider two variants of the GGPT separately. However, in both subsections we will consider a public key generator matrix of the form
$$\hat{G}_{\textnormal{pub}} := S[X \mid G]\sigma \in \F_{q^m}^{k\times (n+\hat t)},$$
where $G$ is a generator matrix of some Gabidulin code $\Gab_{n,k}(\boldsymbol{\alpha})$, $X \in \F_{q^m}^{k\times \hat t}$ is a matrix of column rank $\hat t$, $S \in \GL_k(\F_{q^m})$, and $\sigma \in \GL_{n+\hat t}(\F_q)$.

\subsection{Smart Approach Variant}\label{ss:SA}

Recall from Subection \ref{ss:Overbeck} that we can put the extended matrix into the form 
$$ G_{\textnormal{ext}}' = \tilde{S}'\left(\begin{array}{c|c} X^* & G^* \\ X^{**} & 0\end{array}\right) \sigma .$$
 Rashwan et al.\ in \cite{Rashwan10} proposed what they call the \textit{Smart Approach (SA)}. In this setting, they note that if $X \in \F_{q^m}^{k\times \hat t}$ is constructed from a Moore matrix of column rank $a$ and a non-Moore component of column rank $\hat t-a$, then $X^{**}$ will have rank $\hat t-a$. The paper gives no suggestions for design parameters of such a system. However, one implicit advantage of this construction is the ability to predict the rank of $X^{**}$, and therefore to be able to reduce the public key by choosing a smaller designed error matrix.

In the SA variant we can write $X = X_{\textnormal{Moore}} + Z$ as a minimal column rank Moore decomposition, where $X_{\textnormal{Moore}}$ has column rank $a$ and $Z$ has column rank $\hat t-a$. We can then rewrite 
\begin{equation}\label{eq:SAdecomp} 
\hat{G}_{\textnormal{pub}} =  \underbrace{S[X_{\textnormal{Moore}} \mid G]\sigma}_{M} +  \underbrace{S[Z\mid 0]\sigma}_{X'} .\end{equation} 
$X'$ is a matrix of column rank $\hat t-a$ and $S^{-1}M$ is a Moore matrix generating a code with minimum rank distance at least $n-k+1$.

\begin{theorem}\label{t:SA}
 Consider a GGPT cryptosystem as defined above. 
	Suppose an adversary can find a matrix $U' \in \F_q^{(\hat{t}-a)\times (\hat{t}+n)}$, such that  $\langle U'\rangle = \supp_{\Gr}([Z\mid 0]\sigma) = \supp_{\Gr}(X')$. Then an encrypted message from a public key of the form \eqref{eq:kpubmod} can be recovered in polynomial time.
\end{theorem}

\begin{proof}
We note that $U'$ must be of the form $U' = [U \mid 0]\sigma$, where $U\in \F_q^{(\hat{t}-a)\times \hat{t}}$ is such that $\supp_{\Gr}(Z) = \langle U\rangle$. Let $H_U \in \F_q^{a\times \hat{t}}$ be a parity check matrix for $U$. A parity check matrix $H_{U'}\in \F_q^{(a+n)\times (\hat{t}+n)}$ for $U'$ must be of the form 
$$(H_{U'})^T = \sigma^{-1}\left[\begin{array}{c|c} H_U^T & 0_{\hat{t}\times n} \\ \hline 0_{n\times a} & I_n\end{array}\right]A,$$ for some $A \in \GL_{n+a}(\F_q)$.

We compute
\begin{align*} \hat{G}_{\textnormal{pub}}(H_{U'})^T &= S[X_{\textnormal{Moore}} \mid G]\sigma (H_{U'})^T +  S[Z\mid 0]\sigma (H_{U'})^T \\ &= S[X_{\textnormal{Moore}}\mid G]\left[\begin{array}{c|c} H_U^T & 0_{\hat{t}\times n} \\ \hline 0_{n\times a} & I_n\end{array}\right]A \\ &= S[X_{\textnormal{Moore}}H_U^T \mid G]A.\end{align*} 
 $[X_{\textnormal{Moore}}H_U^T \mid G]A$ is again a Moore matrix, generating a code of minimum distance at least $n-k+1$. Hence, we can find $n$ independent columns of $\hat{G}_{\textnormal{pub}}(H_{U'})^T$ which will form a Gabidulin code of minimum distance $n-k+1$. Denote these columns by $\boldsymbol{i} = (i_1, \ldots i_n)$ and the corresponding submatrix by $G_{\boldsymbol{i}}$. From Section \ref{ss:decoding}, we can recover a decoding algorithm for $\langle G_{\boldsymbol{i}}\rangle$ with respect to $G_{\boldsymbol{i}}$. 
 
 We note that if $\boldsymbol{e}$ is an error of rank at most $t'$, and we denote by $\boldsymbol{e'}$ the subvector of $\boldsymbol{e}(H_{U'})^T$ corresponding to columns $\boldsymbol{i}$, then $$\rk(\boldsymbol{e'}) \leq \rk(\boldsymbol{e}(H_{U'})^T) \leq \rk(\boldsymbol{e}) \leq t'.$$ 
 If we apply $H_{U'}$ to an encrypted message of the form $ \boldsymbol{m}\hat{G}_{\textnormal{pub}} + \boldsymbol{e},$ we obtain 
 $$ \boldsymbol{m}\hat{G}_{\textnormal{pub}}(H_{U'})^T + \boldsymbol{e}(H_{U'})^T.$$ 
 Restricting to the coordinates $\boldsymbol{i}$, we obtain $$\boldsymbol{m}G_{\boldsymbol{i}} + \boldsymbol{e}',$$  which we can decode in the code $\langle G_{\boldsymbol{i}}\rangle $ to recover $\boldsymbol{m}$, since the error correction capability of $\langle G_{\boldsymbol{i}}\rangle $ is $t' \geq \rk(\boldsymbol{e'})$.
\end{proof}

\begin{corollary}
 Consider a GGPT cryptosystem as defined above. If $$\hat{t}-a <  \dfrac{n-k-1}{2}$$
 we can recover an encrypted message in polynomial time.
 \end{corollary}

\begin{proof}
	Recall that $$\hat{G}_{\textnormal{pub}} = \underbrace{S[X_{\textnormal{Moore}} \mid G]\sigma}_{M} +  \underbrace{S[Z\mid 0]\sigma}_{X'}$$
	is a minimum column rank Moore decomposition. Then $S^{-1}\hat{G}_{\textnormal{pub}}  = S^{-1}M + S^{-1}X'$ is also a minimal column rank Moore decomposition. $S^{-1}M$ is a Moore matrix generating a code of minimum rank distance at least $n-k+1$. Since, by the condition of this corollary,
	$$n-k+1 \geq 2(\hat{t}-a)+2,$$ 
	it follows from Corollary \ref{c:rankone} that all elements of rank one in 
$$\sum_{i=0}^{\hat{t}-a} \langle \hat{G}_{\textnormal{pub}} \rangle^{([i])}   =\sum_{i=0}^{\hat{t}-a} \langle S^{-1} \hat{G}_{\textnormal{pub}} \rangle^{([i])}   =\sum_{i=0}^{\hat{t}-a}\langle S^{-1}M+S^{-1}X'\rangle^{([i])}$$ span the space $\supp_{\Gr}(X') = \supp_{\Gr}([Z\mid 0]\sigma)$. 
We can use Lemma \ref{l:rankone} to find these elements of rank one, and obtain $U'\in \F_q^{(\hat{t}-a)\times (\hat{t}+n)}$ such that $\langle U'\rangle = \supp_{\Gr}(X')$. Then we can use Theorem \ref{t:SA} to recover the message. 
\end{proof}

The following example illustrates a case when Overbeck's attack fails, but our attack recovers the encrypted message.

\begin{example}
	Let $q=2$, $n = 8$, $k=3$, $\hat t = 3$, $a = 1$ and $g_1,\dots,g_8\in \F_{2^8}$ linearly independent over $\F_2$. Consider the generator matrix of a Gabidulin code
$$G = \left(\begin{array}{cccc} g_1 & g_2 & \ldots & g_8 \\ g_1^{([1])} & g_2^{([1])} & \ldots & g_8^{([1])} \\ g_1^{([2])} & g_2^{([2])} & \ldots & g_8^{([2])} \end{array}\right),$$
and, for some $x\in \F_{2^8}\backslash \F_2$, the matrices
$$X_{\textnormal{Moore}} = \left(\begin{array}{ccc} x & 0 & 0 \\ x^{([1])} & 0 & 0 \\ x^{([2])} & 0 & 0 \end{array}\right), \quad Z=\left(\begin{array}{ccc} 0 & 1 & 1 \\ 1 & 0 & 1 \\ 1 & 1 & 0 \end{array}\right).$$
Let
$$X = X_{\textnormal{Moore}} + Z$$
and the public key generator matrix be
$$ \hat{G}_{\textnormal{pub}} = [X \mid G] =  \left(\begin{array}{ccc|cccc} x&1&1&g_1 & g_2 & \ldots & g_8 \\ x^{([1])}+1&0&1& g_1^{([1])} & g_2^{([1])} & \ldots & g_8^{([1])} \\ x^{([2])}+1& 1& 0 & g_1^{([2])} & g_2^{([2])} & \ldots & g_8^{([2])} \end{array}\right) .$$
For simplicity we let $S=I_3$ and $\sigma = I_{11}$. We choose $u=1$ and construct $G_{\textnormal{ext}}$, which can be put in the form 
$$G'_{\textnormal{ext}} = \left(\begin{array}{ccc|cccc} x & 0 & 0 & g_1 & g_2 & \ldots & g_8 \\ x^{([1])} & 0 & 0 & g_1^{([1])} & g_2^{([1])} & \ldots & g_8^{([1])}\\ x^{([2])} & 0 & 0 & g_1^{([2])} & g_2^{([2])} & \ldots & g_8^{([2])} \\ x^{([3])} & 0 & 0& g_1^{([3])} & g_2^{([3])} & \ldots & g_8^{([3])} \\ 1 & 0 & 1 & 0& \ldots & 0 & 0 \\ 0 & 1 & 1 & 0 &\ldots  & 0 & 0\end{array}\right)=\left(\begin{array}{c|c} X^* & G^* \\ X^{**} & 0\end{array}\right)$$
by a suitable row transformation. 
Here Overbeck's attack fails, because $X^{**}$ does not have full rank.
On the other hand, our attack succeeds, since we can directly recover the elements of rank one as $\langle [X^{**}\mid 0]\rangle=\langle [Z\mid 0] \rangle$. Thus we can use Theorem \ref{t:SA} and recover any encrypted message.
\end{example}

\subsection{Loidreau's GGPT Variant}\label{ss:LGGPT}

As already mentioned in Subsection \ref{ss:Overbeck}, in Loidreau's GGPT  variant \cite{Loidreau10} the designed error matrix, $X \in \F_{q^m}^{k\times \hat t}$, is a randomly chosen matrix of rank $a< \hat{t}/(n-k)$. We can assume that $X$ has column rank $\hat t$. 
In this case, to find the Grassmann support of $X$ with the help of Theorem \ref{t:frobchain}, we need to go up to the $(\hat t-1)$-st Frobenius power. But, since 
$$\hat t-1> a(n-k)-1 \geq n-k-1,$$ 
we get that 
$$\sum_{i=0}^{\hat t-1} G^{([i])} = \sum_{i=0}^{n-k-1} G^{([i])} =\F_{q^m}^n , $$
hence the elements of rank one cannot help us reconstruct the Grassmann support of $X$. Thus the attack of Subsection \ref{ss:SA} would not succeed. 

However, in this case, we can still use the idea of locating the elements of rank one; but now we want to recover the elements of the Gabidulin part of the code, instead of the Grassmann support of $X$. The strategy is effectively the same, although we must make some assumptions on the behavior of $X$ and random subcodes of Gabidulin codes. 


First, we note that there is a suitable row transformation, $T$, so that \begin{equation}\label{eq:form}T\hat{G}_{\textnormal{pub}} = \left(\begin{array}{c|c} X^* & G^* \\ 0 & G^{**} \end{array}\right)\sigma,\end{equation} where $X^* \in \F_{q^m}^{a\times \hat{t}}$ is a matrix with the same row span as $X$, and $G^*$ and $G^{**}$ are matrices which span subcodes of $\langle G\rangle$. One can easily see that $\langle [X^*\mid G^*]\sigma\rangle $ and $\langle [0\mid G^{**}]\sigma\rangle $ intersect trivially. Our strategy will be to use the purely Gabidulin part $[0\mid G^{**}]$ to generate a parity check matrix for $[X^* \mid 0]$. 

We will now state the assumptions that we will use in our attack. These assumptions are justified with experimental results in Table \ref{table:assumptions}.

\begin{assumption}\label{ass:1}
	Let $G \in \F_{q^m}^{k\times n}$ be a generator matrix of a Gabidulin code, and $\mathcal{B}\subset\langle G\rangle$ be a random subspace of $\langle G\rangle$ of codimension $a$. Set \begin{equation}\label{eq:ell} \ell = \left\lceil \dfrac{n}{k-a}\right\rceil.\end{equation} With high probability, we have \begin{equation}\label{eq:wholespace}\sum_{i=0}^{\ell-1}\mathcal{B}^{([i(k-a)])} = \F_{q^m}^n.\end{equation}
\end{assumption}

The value of $\ell$ in \eqref{eq:ell} is the smallest possible value for which we can obtain equality in \eqref{eq:wholespace}. One could choose $\ell$ larger than in \eqref{eq:ell}, which we will remark on in the end of this section. 

Since we can expect a random matrix whose rank is small relative to the dimension of the ambient space to not contain elements of rank one, we make the additional assumption:
\begin{assumption}\label{ass:2}
	Let $X \in \F_{q^m}^{k\times \hat{t}}$ be a random matrix of rank $a$. For $\ell$ given in \eqref{eq:ell}, if $\ell a \ll \hat{t}$, then with high probability, $$\sum_{i=0}^{\ell-1}\langle X\rangle^{([i(k-a)])}$$ contains no elements of rank one.
\end{assumption}

\begin{table}[t]
\centering
\begin{tabular}{|ccccc|c|c|} \hline $m$ & $n$ & $k$ & $a$ & $\hat{t}$ & $\textnormal{Assumption } \ref{ass:1} $ & $\textnormal{Assumption } \ref{ass:2}$ \\ \hline $24$ & $24$ & $12$ & $3$ & $40$ & $\sim 1$ & $\sim 1$ \\ $24$ & $24$ & $12$ & $4$ & $52$ & $\sim .998$ & $\sim 1$ \\\hline \end{tabular}
\caption{Experimental results for Assumptions 1 and 2: Probabilities of success in $1000$ trials for $q = 2$.}
\label{table:assumptions} 
\end{table}

\begin{theorem}
Let $S \in \GL_k(\F_{q^m})$, $\sigma\in\GL_{n+\hat{t}}(\F_q)$, $G \in \F_{q^m}^{k\times n}$, and $X \in \F_{q^m}^{k\times \hat{t}}$ be of rank $a$, and consider Loidreau's GGPT variant with public key 
$$\hat{G}_{\textnormal{pub}} = S[X\mid G]\sigma.$$ 
If Assumptions \ref{ass:1} and \ref{ass:2} are true, then we can break the Loidreau GGPT variant in polynomial time with high probability.
\end{theorem}

\begin{proof}
	Let $\ell$ be as in $\eqref{eq:ell}$ and $T\hat G_{\textnormal{pub}}$ as in $\eqref{eq:form}$. Consider the  matrix 
\begin{equation*}\label{eq:pubmatrixlggpt} G''_{\textnormal{ext}} :=\left(\begin{array}{c} \hat{G}_{\textnormal{pub}} \\ \hat{G}_{\textnormal{pub}}^{([k-a])} \\ \vdots \\ \hat{G}_{\textnormal{pub}}^{([(k-a)(\ell-1)])}\end{array}\right) = \tilde{S} \underbrace{\left(\begin{array}{c} (X^*\mid G^*) \\ (X^*\mid G^*)^{([k-a])} \\ \vdots \\ (X^*\mid G^*)^{([(k-a)(\ell-1)])} \\ \hline (0\mid G^{**}) \\ (0\mid G^{**})^{([k-a])} \\ \vdots \\ (0\mid G^{**})^{([(k-a)(\ell-1)])}\end{array}\right)}_{\bar{G}}\sigma.\end{equation*} 
Since $\langle G^{**}\rangle$ is a subcode of $\langle G\rangle$ of codimension $a$, by Assumption \ref{ass:1}, we have with high probability,
$$\sum_{i=0}^{\ell-1}\langle G^{**}\rangle^{([(k-a)i])} = \F_{q^m}^n.$$ 
Then, the bottom submatrix of $\bar{G}$ has the same row span as $[0 \mid I_n]$, and hence, by using elementary operations, we can eliminate the second component of every row in the top submatrix of $\bar{G}$. Then, the space generated by the rows of $G''_{\textnormal{ext}}$ is the same as that generated by 
$$G'''_{\textnormal{ext}} = \left(\begin{array}{c|c} X^* & 0 \\ (X^*)^{([k-a])} & 0 \\ \vdots & \vdots \\(X^*)^{([(k-a)(\ell-1)])} & 0 \\ 0 & I_n \end{array}\right)\sigma.$$ 
By Assumption \ref{ass:2}, with high probability we have that $$\sum_{i=0}^{\ell-1}\langle X^*\rangle^{([(k-a)i])}$$ contains no elements of rank one, and therefore all elements of rank one in $\langle G''_{\textnormal{ext}}\rangle=\langle G'''_{\textnormal{ext}}\rangle$ must belong to $\langle [0\mid I_n]\rangle\sigma$. With the help of Lemma \ref{l:rankone} we can recover a matrix $U \in \F_q^{n\times (\hat{t}+n)}$ which is a basis for $\langle [0\mid I_n]\rangle\sigma$. Then, any parity check matrix $H_U \in \F_q^{\hat{t}\times (n+\hat{t})}$ for $\langle [0\mid I_n]\rangle\sigma$ must have the form $$H_U^T = \sigma^{-1}\left[\begin{array}{c} A \\ 0\end{array}\right] \in \F_q^{(n+\hat{t})\times \hat{t}},$$ where $A \in \GL_{\hat{t}}(\F_q)$. It follows that, if we compute 
$$\hat{G}_{\textnormal{pub}}H_U^T = S[X \mid G] \sigma H_U^T =SXA \in \F_{q^m}^{k\times \hat{t}},$$ 
then there exists a unique matrix $V = [A^{-1} \mid 0]\sigma \in \F_q^{\hat{t}\times(n+\hat{t})}$ such that $$\hat{G}_{\textnormal{pub}}H_U^T V = SXAV = S[X\mid 0]\sigma.$$ We can find the matrix $V$ by observing that 
\begin{equation}\label{eq:LGGPTeq} (\hat{G}_{\textnormal{pub}} - \hat{G}_{\textnormal{pub}}H_U^T V)H_U^T = S[0\mid G]\sigma H_U^T = 0.\end{equation} 
This gives a linear system of equations with $\hat{t}(n+\hat{t})$ variables and $k\times \hat{t}$ equations over $\F_{q^m}$. Since the variables can only take values in $\F_q$, we can expand each equation into $m$ equations over $\F_q$, obtaining a system of $km\hat{t}$ equations and $\hat{t}(n+\hat{t})$ variables over $\F_q$. Hence, we can solve this system of equations if $km \geq n+\hat{t}$ (which is always satisfied).

Let $H_V \in \F_{q^m}^{n\times (n+\hat{t})}$ be any dual matrix for $V$. Then, $H_V$ has the form $$H_V^T = \sigma^{-1}\left[\begin{array}{c} 0 \\ B \end{array}\right],$$ 
for some $B \in \GL_n(\F_q)$. Therefore, 
$$\hat{G}_{\textnormal{pub}}H_V^T = S[X\mid G]\sigma H_V^T = SGB \in \F_{q^m}^{k\times n}$$
 is a Gabidulin code of minimum distance $n-k+1$, from which we can recover a decoding algorithm, as explained in Subsection \ref{ss:decoding}. If we receive an encrypted message of the form 
 $$\boldsymbol{m}\hat{G}_{\textnormal{pub}} + \boldsymbol{e},$$ we can apply $H_V$, obtaining 
 $$\boldsymbol{m}SGB + \boldsymbol{e}H_V^T.$$ Since  
 $$\rk(\boldsymbol{e}H_V^T) \leq \wt_{\textnormal{R}}(\boldsymbol{e}) \leq t',$$ we can recover the encrypted message, $\boldsymbol{m}$, from the recovered decoding algorithm with respect to $SGB$. All the operations required for this attack can be performed in polynomial time.
\end{proof}

We will conclude this section with an example where we analyze our attack against the parameters proposed by Loidreau in \cite{Loidreau10} in order to resist Overbeck's attack. It turns out that the proposed parameters are not secure against our attack. 

\begin{example}
Consider a Loidreau GGPT variant with $q = 2$, $m=n=24$, $k=12$, $a=3$, and $\hat{t} = 40$, i.e.\ the first set of parameters from Table \ref{table:assumptions}. 
Assume we, as an attacker, know the public generator matrix $\hat{G}_{\textnormal{pub}}\in \F_{2^{24}}^{12\times 64}$ and received an encrypted message $\boldsymbol{y}$.
We compute $\ell = \lceil \frac{24}{12-3}\rceil = 3$ and proceed as follows:
\begin{enumerate}
	\item We compute $\hat{G}_{\textnormal{pub}}^{([9])}, \hat{G}_{\textnormal{pub}}^{([18])}$ to obtain the extended matrix $G_{\textnormal{ext}}'' \in \F_{2^{24}}^{36\times 64}$. This requires at most $1536 = 2\cdot 12\cdot 64$ Frobenius powers in $\F_{2^{24}}$. Using a normal basis to represent $\F_{2^{24}}$ over $\F_2$, this can be done very efficiently. 
	\item\label{step3} We find the elements of rank one in $\langle G_{\textnormal{ext}}''\rangle$, as described in Lemma \ref{l:rankone}. To do so we need to row reduce $G_{\textnormal{ext}}''$ and then solve a linear system over $\F_2$ with  
$36$ unknowns and $24\cdot 64 = 1536$ equations. Then, if Assumptions 1 and 2 hold, we find some basis matrix $U \in \F_2^{24\times 64}$, such that $\langle U\rangle$ contains all these elements of rank one.
	\item Compute a parity check matrix $H_U$ for $U$.
	\item\label{step4} We find a matrix $V\in \F_2^{40 \times 64}$, solving Equation \eqref{eq:LGGPTeq}.
	\item We compute a parity check matrix $H_V \in \F_2^{64\times 24}$  for $V$, and compute the product $\hat{G}_{\textnormal{pub}}H_V^T$.
	\item We recover a decoding algorithm for the code $\langle \hat{G}_{\textnormal{pub}}H_V^T\rangle $, as described in Lemma \ref{l:rankone}, and decode $\boldsymbol{y}H_V^T$ with this algorithm.
\end{enumerate}

We observe that step \ref{step4} above is the most computationally intensive, and therefore we estimate the complexity of our attack based on this step. This is done by solving a $(40\cdot 64)\times (24\cdot 12\cdot 40)$ system over $\F_2$ by Gaussian elimination on the resulting matrix. This requires on the order of $2^{39}$ operations over $\F_2$. Implementing the algorithm on a personal computer, we were able to break this system very efficiently.

For the parameters in the second row of the table, we can similarly break the system, albeit with slightly higher complexity due the larger parameters.
\end{example}

\section{Conclusion}

In this paper, we provide a new attack against cryptosystems based on Gabidulin codes, reconfirming that Gabidulin based cryptosystems are vulnerable from a structural perspective. Our attack generalizes Overbeck's attack, focusing instead on recovering the elements of rank one, rather than the structure of the dual space. 
One principle advantage of our attack is that it can be extended to cryptanalyze certain variants of the generalized GPT system, which resist the original attacks of Gibson and Overbeck. In particular, we show that the Smart Approach and Loidreau's GGPT variants are vulnerable to this attack. 

To the best of the authors' knowledge, attacking a cryptosystem by looking at the elements of rank one is a new approach which may need to be considered for the security of existing and future rank-metric based cryptosystems. As a next step the authors want to use this idea of finding elements of rank one to cryptanalyze the column scrambler variant of \cite{Gabidulin09}, which has so far resisted structural attacks.

\bibliography{biblio}{}
\bibliographystyle{plain}

\end{document}